\documentclass[10pt]{article}

\usepackage{amssymb, amsmath}
\usepackage{amsthm}
\usepackage{graphics}
\usepackage{amsfonts}
\usepackage{mathrsfs}
\usepackage{amscd}
\usepackage{comment}
\usepackage{color}
\usepackage{url}
\newtheorem{thm}{Theorem}
\newtheorem{lem}[thm]{Lemma}

\newtheorem{prop}[thm]{Proposition}
\newtheorem{define}[thm]{Definition}

\newtheorem{fact}[thm]{Fact}

\newtheorem{example}[thm]{Example}

\newcommand{\vx} {{\bf x}}

\newcommand{\bE} { {E}}

\newcommand{\gl} { {\rm {gl}}}

\newcommand{\lrD} { {\langle D_t \rangle}}

\newcommand{\ldr} { {\langle D_t, D_{1}, \ldots, D_{n}\rangle}}

\newcommand{\Gal} { \text{Gal}}
\newcommand{\ann} { \operatorname{ann}}

\newcommand{\bK} { {K}}
\newcommand{\bbR} { {\mathbf{R}}}

\newcommand{\pa}{\partial}

\begin{document}

\title{Parallel Telescoping and Parameterized Picard--Vessiot Theory
\footnote{S.C.\, R.F.\ and Z.L. were supported
by two NSFC grants (91118001, 60821002/F02) and a 973 project (2011CB302401),
M.F.S. was supported by the NSF grant CCF-1017217. }
}

\author{Shaoshi Chen, Ruyong Feng, and Ziming Li\\
KLMM, AMSS, Chinese Academy of Sciences\\Beijing 100190, (China)\\ {\tt {schen,ryfeng}@amss.ac.cn},\\{\tt zmli@mmrc.iss.ac.cn} \\\\
Michael F.\ Singer\\
Department of Mathematics,\\ North Carolina State University \\Raleigh, NC 27695-8205 (USA)\\
{\tt singer@math.ncsu.edu}
}

\maketitle

\begin{abstract}
Parallel telescoping is a natural generalization of differential creative-telescoping
for single integrals to line integrals. It computes a linear ordinary differential
operator~$L$, called a parallel telescoper, for several multivariate functions, such that
the applications of~$L$ to the functions yield antiderivatives of a single function.
We present a necessary and sufficient condition guaranteeing  the existence of parallel telescopers
for differentially finite functions, and develop an algorithm to compute minimal ones for compatible
hyperexponential functions.
Besides computing annihilators of parametric line integrals, we use the parallel telescoping for
determining Galois groups of parameterized partial differential systems of first order.
\end{abstract}

\section{Introduction}\label{SECT:intro}
The problem of finding linear differential equations with polynomial
coefficients for parametric  integrals has a long history. It at least dates back to Picard~\cite{Picard1902} who proved
the existence of such linear differential equations for
integrals of algebraic functions involving parameters.
His result has been generalized
to higher-dimensional cases and led to
Gauss--Manin connections~\cite{Manin1958, Manin1963, KatzOda1968}.
The key for obtaining such linear differential equations is the method of \emph{creative telescoping},
which was first formulated as an algorithmic tool by Zeilberger and his collaborators
in 1990s~\cite{Zeilberger1990, Zeilberger1991, Wilf1992}. The method enables us to
prove a large amount of combinatorial identities in an automatic way~\cite{PWZbook1996}.
For more recent developments, see the survey article~\cite{Koutschan2013}.

Given a function~$f(t, x)$ described by two linear differential equations with polynomial coefficients in~$t$ and~$x$,
the method of differential creative-telescoping \cite{Almkvist1990} finds
a linear differential operator~$L$ in~$\pa/\pa t$ with polynomial coefficients in~$t$ such that
$L(f) = \pa g/\pa x,$
where~$g$ is usually a linear combination of partial derivatives of~$f$ over the field of rational functions in~$t$ and~$x$.
The operator~$L$ is called a \emph{telescoper} for~$f$, and
the function~$g$ is called a \emph{certificate} of~$L$.
They can be used to evaluate parametric integrals of~$f$ with respect to~$x$.

Recently, a connection has been revealed between the method of differential creative-telescoping and
Galois theory of parameterized differential equations in~\cite{CassidySinger2007, Arreche2012, Dreyfus2011, Singer2013, Gorchinskiy2013}.
Consider a first-order partial differential system of the form:
\begin{equation} \label{EQ:pps}
\frac{ \pa Y}{\pa x_1}  = f_1, \, \,\,   \ldots, \,\,  \, \frac{\pa Y}{\pa x_n}  = f_n,
\end{equation}
where~$f_1, \ldots, f_n$ are rational functions in~$t$, $x_1,\ldots, x_n$ satisfying compatibility conditions.
Its parameterized Galois group can be determined by
constructing a linear ordinary differential operator~$L$ in~$\pa/\pa t$ with polynomial coefficients in~$t$ such that
\[L(f_1) = \frac{\pa g}{\pa x_1} , \ldots, L(f_n)= \frac{\pa g}{\pa x_n} \]
for a single rational function~$g$. The operator~$L$ will be referred as to a parallel telescoper
for~$f_1$, \ldots,~$f_n$ with respect to~$x_1,$ \ldots,~$x_n$.
Parallel telescopers
may also be used to evaluate parametric line integrals
in the same manner as we do for single integrals by classical creative-telescoping.

In this paper, we present a necessary and sufficient condition guaranteeing  the existence of parallel telescopers
for differentially finite functions (see Definition~\ref{DEF:d-finite}).  The condition can easily be verified if the given functions are hyperexponential.
We develop an algorithm to compute a parallel telescoper of minimal order for compatible hyperexponential functions.
The algorithm can be used for constructing parallel telescopers for non-compatible ones, although its output may not be of minimal order.
We also show how to determine the Galois group of a differential system of the form~\eqref{EQ:pps} by parallel
telescoping.

The rest of the paper is organized as follows. In Section~\ref{SECT:finite}, we review
the notion of differentially finite elements.
In Section~\ref{SECT:paratele},
we study the existence of parallel telescopers. We present an algorithm in Section~\ref{SECT:hyper}
for constructing minimal parallel telescopers for hyperexponential functions.
In Section~\ref{SECT:ppv},
parallel telescoping is applied
to determine Galois groups of parameterized partial differential
systems of first order.

\section{Differentially finite elements}\label{SECT:finite}
Let~$k$ be an algebraically closed field of characteristic zero. Assume that~$\delta_i$
is the usual partial derivative with respect to~$x_i$ on the field~$k(x_1, \ldots, x_n)$ for all~$i$
with~$1 \le i \le n$.
For brevity, we set~$\vx=(x_1, \ldots, x_n)$.
Over the differential field~$\left(k(\vx),  \{\delta_1, \ldots, \delta_n\}\right)$
there is a noncommutative algebra~$k(\vx)\langle D_1, \ldots, D_n\rangle$ whose commutation rules are
\[ D_i D_j = D_j D_i \quad \text{and} \quad
D_i f = f D_i + \delta_i(f) \]
for all~$i,j \in \{1, \ldots, n\}$ and~$f \in k(\vx)$. The algebra is also called the ring of differential
operators associated to~$k(\vx)$. The commutation rules imply the following fact:
\begin{fact} \label{FACT:bracket}
Let~$L \in k(\vx)\langle D_1, \ldots, D_n\rangle$ and~$[D_i, L]=D_i L - L D_i$
for some~$i$ with~$1 \le i \le n$.
\begin{enumerate}
\item[(i)] $[D_i, L]=0$ if and only if~$L$ is free of~$x_i$.
\item[(ii)] If~$L$ is free of~$D_i$, then so is~$[D_i, L]$.
\item[(iii)] If~$L$ is in~$k[\vx]\langle D_1, \ldots, D_n\rangle$, then the degree of~$[D_i,L]$
in~$x_i$ is less than that of~$L$.
\end{enumerate}
\end{fact}

Due to the noncommutativity of~$k(\vx)\langle D_1, \ldots, D_n\rangle$, we make a convention
that ideals, vector spaces, modules and submodules are all left ones in this paper.

Let~$M$ be a module over~$k(\vx)\langle D_1, \ldots, D_n\rangle$.
For an operator $L$ in the ring~$k(\vx)\langle D_1, \ldots, D_n\rangle$ and~$h \in M$,
the scalar product of~$L$ and~$h$ is denoted by~$L(h)$. We say that~$L$ is an annihilator
of~$h$ if~$L(h)=0$. The set of all annihilators of~$h$ is denoted by~$\ann(h)$,
which is an ideal in~$k(\vx)\langle D_1, \ldots, D_n \rangle$.

\begin{define}\label{DEF:d-finite}
Let~$h$ be an element of a module over the ring~$k(\vx)\langle D_1, \ldots, D_n\rangle$.
We say that~$h$ is \emph{differentially finite} (abbreviated as $D$-finite) over~$k(\vx)$
if
\[\text{$\ann(h) {\cap} k(\vx)\langle D_i \rangle {\neq} \{0\}$
for all~$i$ with~$1 \le i \le n$}.\]
\end{define}
It is straightforward to see that~$h$ is $D$-finite if and only if the
submodule generated by~$h$ is a finite-dimensional linear space over~$k(\vx)$. It follows that,
if~$h_1, \ldots, h_m$ are $D$-finite elements in a module over~$k(\vx)\langle D_1, \ldots, D_n\rangle$,
so is every element in the submodule generated by~$h_1, \ldots, h_m$.

When a  module consists of functions in~$x_1, \ldots, x_n$,
its $D$-finite elements are called $D$-finite functions, which are ubiquitous in combinatorics as generating functions.
$D$-finite functions were
first systematical investigated by Stanley in~\cite{Stanley1980}. Their important algebraic
properties have been  revealed by Lipshitz in~\cite{Lipshitz1988, Lipshitz1989}.
We recall a lemma in~\cite[Lemma 3]{Lipshitz1988}, which is the starting point of our study on parallel telescopers.

%
%
%
%
%
%

\begin{lem}[Lipshitz, 1988]\label{LM:lipshitz}
If~$h$ is a $D$-finite element in a module over the ring $k(\vx)\langle D_1, \ldots, D_n \rangle$,
then
\[ \ann(h) \cap  k(x_1, \ldots, x_{i-1}, x_{i+1}, \ldots, x_n)\langle D_i, D_j\rangle \neq \{0\} \]
for all~$i, j\in \{1, \ldots, n\}$ with~$i\neq j$.
\end{lem}

%

The next lemma allows one to remove redundant variables.
\begin{lem} \label{LM:eliminate}
Let~$h$ be a $D$-finite element in a module over~$k(\vx)\langle D_1, \ldots, D_n \rangle$.
If
$$D_{m+1}(h)= D_{m+2}(h)= \cdots = D_n(h) =0$$
 for some~$m \in \{1, \ldots, n-1\}$, then~$h$ is also a $D$-finite
element
over~$k(x_1, \ldots, x_m)$.
\end{lem}
\begin{proof}
By the definition of $D$-finite elements, it suffices to show that
the intersection of~$\ann(h)$ and~$k(x_1, \ldots, x_m)\langle D_i \rangle$ is nontrivial
for all~$i$ with~$1 \le i \le m$.
Suppose the contrary. Then, without loss of generality, we may further suppose
that every nonzero annihilator of~$h$ in~$k[\vx]\langle D_1 \rangle$ involves~$x_n$.
Among those annihilators, we choose one, say~$P$, whose degree in~$x_n$ is minimal.
By Fact~\ref{FACT:bracket}~(i),~$[D_n, P]$ is nonzero. By Fact~\ref{FACT:bracket}~(ii), it belongs
to~$k(x_1, \ldots, x_n)\langle D_1 \rangle$.
Since both~$D_n(h)$ and~$P(h)$ are equal to zero,~$[D_n,P]$ is also a nonzero annihilator of~$h$
in~$k[\vx]\langle D_1 \rangle$. By Fact~\ref{FACT:bracket}~(iii), it has degree in~$x_n$ less than that of~$P$, a contradiction.
\end{proof}

%
%

\section{Parallel Telescopers}\label{SECT:paratele}
In this section, we define the notion of parallel telescopers for several multivariate functions
in a module-theoretic setting, and study under what conditions parallel telescopers exist
for $D$-finite elements.

\subsection{Definition of parallel telescopers} \label{SUBSECT:defn}
In order to define parallel telescopers, we introduce a new indeterminate~$t$,
and extend the field~$k(\vx)$ to~$k(t, \vx)$, which is denoted by~$\bK$,
and set $\Delta= \{\delta_t, \delta_1, \ldots, \delta_n\}$, where~$\delta_t$ stands for the usual partial derivative
with respect to~$t$ on~$\bK$. Moreover,
let us denote by~$\bbR$ the ring~$\bK\ldr$  of linear differential operators.
The notions such as $D$-finite elements and annihilators carry over naturally to~$\bK$ and~$\bbR$.
\begin{define} \label{DEF:paratele}
Let $f_1, \ldots, f_n$ be in an $\bbR$-module.
A nonzero operator~$L\in k(t)\lrD$ is called a \emph{parallel telescoper} for~$f_1,$ \ldots, $f_n$
with respect to~$\vx$ if~there exists an element~$g$ in the submodule generated by~$f_1$, \ldots,~$f_n$ over~$\bbR$,
such that
$$L(t, D_t)(f_i) = D_{i}(g) \quad \mbox{for all~$1\leq i \leq n$}.$$
{The element~$g$ is called a \emph{certificate} of~$L$ with respect to~$\vx$.}
\end{define}
By Definition~\ref{DEF:paratele},
the  parallel telescopers for~$f_1$, \ldots,~$f_n$ and zero form an ideal in the ring~$k(t)\lrD$.
The ideal is principal, since~$k(t)\lrD$ is a left Euclidean domain.
A generator of the ideal is called  a \emph{minimal parallel telescoper} for~$f_1, \ldots, f_n$ with respect to~$\vx$.

\subsection{Existence of parallel telescopers}\label{SUBSECT:existence}
We derive a necessary and sufficient condition on the existence of parallel telescopers for $D$-finite elements.
To this end, we need a differential analogue of~\cite[Thm. 6.2.1]{PWZbook1996}.

\begin{lem} \label{LM:tele}
Let~$h$ be an element of an $\bbR$-module. If~$h$ is D-finite over~$\bK$, then, for every~$i\in \{1, \ldots, n\}$, there exists a nonzero
operator
\[L_i \in k(t, x_1, \ldots, x_{i-1}, x_{i+1}, \ldots, x_n)\langle D_t\rangle\]
 and an element~$g_i$ in the submodule
generated by~$h$ such that~$L_i(h) = D_i(g_i).$
\end{lem}
\begin{proof}
Let~$N$ be the submodule generated by~$h$ over~$\bbR$.

By Lemma~\ref{LM:lipshitz}, $h$ has a nonzero annihilator in~$\bK\langle D_t, D_i\rangle$, which
is free of~$x_i$.
Among all of the $x_i$-free and nonzero annihilators for~$h$, we choose one, say~$P_i$, whose degree in~$D_{i}$ is minimal.
If the degree~$d_i$ of~$P_i$ in~$D_{i}$ is equal to zero, then the lemma holds by taking~$L_i = P_i$ and~$g_i=0$.
Assume that~$d_i > 0$. We can always write
\begin{equation} \label{EQ:ldiv}
P_i = L_i + D_{i} Q_i,
\end{equation}
where~$L_i$ is in the ring~$k(t, x_1, \ldots, x_{i-1}, x_{i+1}, \ldots, x_n)\langle D_t\rangle$,
and~$Q_i$ is in the ring $k(t, x_1, \ldots, x_{i-1}, x_{i+1}, \ldots, x_n)\langle D_t, D_i\rangle$
whose degree in~$D_i$ is less than~$d_i$.

Set~$g_i:=Q_i(h)$. Since~$P_i(h)=0$ and $g_i$ is in~$N$,  it remains to show that~$L_i$ is nonzero
in~\eqref{EQ:ldiv}.
Suppose that~$L_i = 0$. Then~$D_i(g_i)=0$, which, together with the $D$-finiteness of~$g_i$, implies that $g_i$ is $D$-finite over~$k(t, x_1, \ldots, x_{i-1}, x_{i+1}, \ldots, x_n)$
by Lemma~\ref{LM:eliminate}. Thus,
there exists a nonzero linear operator~$R_i$ in the ring $k[t, x_1, \ldots, x_{i-1}, x_{i+1}, \ldots, x_n]\langle D_t\rangle$
such that~$R_i(g_i)=0$.
It follows that the product~$R_i Q_i$ is also a nonzero and $x_i$-free annihilator of~$h$. But it has degree in~$D_i$ less
than~$d_i$, which contradicts the minimality assumption for the degree of~$P_i$ in~$D_i$.
\end{proof}
To study the existence of parallel telescopers, we introduce the notion of compatible elements with respect to~$\vx$.
\begin{define} \label{DEF:compatible}
The elements~$f_1, \ldots, f_n$ of an $\bbR$-module are said to be \emph{compatible} with respect to~$\vx$
if the compatibility conditions
$D_{i}(f_j) = D_{j}(f_i)$ for all~$1\leq i < j \leq n$ hold.
\end{define}

The following lemma shows that the compatibility conditions are sufficient for
the existence of parallel telescopers for $D$-finite elements.

\begin{lem}\label{LM:suffcond}
Let $f_1, \ldots, f_n$ be elements of an $\bbR$-module. If they are $D$-finite over~$\bK$
and compatible with respect to~$\vx$,
then there exists a parallel telescoper for~$f_1,$ \ldots,~$f_n$ with respect to~$\vx$.
\end{lem}
\begin{proof}
Set~$\vx_m=(x_1, \ldots, x_m)$,  and set~$\bbR_m :=k(t, \vx_m)\langle D_t, D_1, \ldots, D_m \rangle$
for all~$m$ with~$1 {\le} m {\le} n$.

We proceed by induction on~$n$.
If~$n{=}1$, then~$f_1$ has a telescoper in~$k(t)\langle D_t \rangle$ with respect to~$\vx_1$
by Lemma~\ref{LM:tele}.
Assume that the lemma holds for any~$n-1$ elements that are both $D$-finite over~$k(t, \vx_{n-1})$ and compatible with
respect to~$\vx_{n-1}$.

Assume that~$f_1, f_2, \ldots, f_n$ are $D$-finite over~$k(t, \vx_n)$ and
compatible with respect to~$\vx_n$.
Denote by~$N$ the submodule generated by~$f_1$, \ldots,~$f_n$ over~$\bbR_n$.
By Lemma~\ref{LM:tele},
there exists  a nonzero operator~$L_n$ in~$k(t, \vx_{n-1})\lrD$
such that
\begin{equation} \label{EQ:teleM}
L_n (f_n) = D_n(g_n) \quad \mbox{for some~$g_n \in N.$}
\end{equation}
Without loss of generality,  we further assume that~$L_n$ in~\eqref{EQ:teleM} is of minimal degree in~$D_t$ and is monic with respect to~$D_t$.

First, we show that~$L_n$ belongs to~$k(t)\lrD$.
For all~$i$ with~$1 \le i \le n-1$, we set~$L_i = [D_i, L_n]$,
which belongs to~$k(t, \vx_{n-1})\langle D_t \rangle$ by Fact~\ref{FACT:bracket}~(ii), and has degree in~$D_t$ less than that of~$L_n$,
because~$L_n$ is monic with respect to~$D_t$. Note that~$L_n D_i (f_n) = L_n D_n(f_i) = D_n L_n (f_i)$,
in which the first equality is immediate from the compatibility condition~$D_i(f_n)=D_n(f_i)$, and the second from the fact
that~$L_n$ is free of~$x_n$. Thus,~$L_i(f_n) = D_i L_n(f_n) -  D_n L_n(f_i),$
which, together with~\eqref{EQ:teleM}, implies that
\begin{equation} \label{EQ:freeM}
L_{i}(f_n) = D_i D_n (g_n) - D_n L_n (f_i) = D_n \left( \tilde{f}_i \right),
\end{equation}
where~$\tilde{f}_{i}:=D_i(g_n) - L_n(f_i)$ for~$i=1, \ldots, n-1$.
Since~$\tilde{f}_{i}$  belongs to~$N$, we see that~$L_i=0$, for otherwise,~$L_n$ would not be
a nonzero operator that satisfies~\eqref{EQ:teleM} and has minimal degree in~$D_t$ by~\eqref{EQ:freeM}.
Thus,~$L_n$ is free of~$x_i$ by Fact~\ref{FACT:bracket}~(i).
Accordingly,~$L_n \in k(t)\langle D_t \rangle$.
Moreover,~$L_{i}=0$ and~\eqref{EQ:freeM} imply
\begin{equation} \label{EQ:const}
D_n(\tilde{f}_{i}) =0 \quad \mbox{for all~$i$ with~$1 \le i \le n-1$}.
\end{equation}

Next, we apply the induction hypothesis to~$\tilde{f}_{1}, \ldots, \tilde{f}_{n-1}$.
Since~$f_1, \ldots, f_n$ are $D$-finite over~$k(t, \vx_n)$, so is~$g_n$, and so is~$\tilde{f}_{i}$ for all~$i$ with~$1 \le i \le n-1$. By~\eqref{EQ:const} and
Lemma~\ref{LM:eliminate},~$\tilde{f}_{i}$
is $D$-finite over~$k(t, \vx_{n-1}).$
Moreover,~$\tilde{f}_{1}, \ldots, \tilde{f}_{n-1}$ are compatible with respect to~$\vx_{n-1}$
because~$f_1$, \ldots, $f_{n-1}$
are compatible with respect to~$\vx_{n-1}$ and because~$L_n$ is free of~$\vx_{n-1}$.
Therefore, there exist a nonzero
operator~$\tilde{L} \in k(t)\langle D_t \rangle$ and
an element~$\tilde{g}$  in the submodule generated by~$\tilde{f}_{1}, \ldots, \tilde{f}_{n-1}$
over~$\bbR_{n-1}$
such that
\begin{equation} \label{EQ:pt}
   \tilde{L}\left(\tilde f_{i}\right) {=} D_i\left(\tilde{g}\right)~\mbox{for~$i {\in} \{1, \ldots, n-1\}$}~~\text{and}~~D_n\left(\tilde{g} \right) {=} 0.
\end{equation}
The first equality in~\eqref{EQ:pt} is due to the induction hypothesis, and the second
due to~\eqref{EQ:const}.
Moreover,~$\tilde{g}$ belongs to~$N$.

At last, we verify that~$\tilde{L} L_n$ is a parallel telescoper for $f_1,$ \ldots, $f_n$.
Set~$g= \tilde{L}(g_n)-\tilde{g}$. It belongs to~$N$ because both~$g_n$ and~$\tilde{g}$ do.
For~$i \in \{1, \ldots, n-1\}$,
$\tilde{L} L_n (f_i) = \tilde{L}\left(D_i(g_n) - \tilde f_{i}\right)$
by the definition of~$\tilde f_i$ in~\eqref{EQ:freeM}. It follows from~$\tilde{L}D_i = D_i \tilde{L}$ and the first equality of~\eqref{EQ:pt}
that, for all~$i \in \{1, \ldots, n-1\}$,
\[ \tilde{L} L_n (f_i) = D_i \tilde{L}(g_n) - D_i \left( \tilde{g} \right)
= D_i \left( \tilde{L}(g_n) - \tilde{g}  \right) = D_i(g). \]
Applying~$\tilde{L} L_n$ to~$f_n$, we get
\[ \tilde{L} L_n(f_n) = \tilde{L} D_n(g_n) = D_n \left(\tilde{L}(g_n)\right) = D_n (g), \]
in which the first equality follows from~\eqref{EQ:teleM} and the last from the second one in~\eqref{EQ:pt}. Therefore,~$\tilde{L} L$
is indeed a parallel telescoper for~$f_1, \ldots, f_n$ with respect to~$\vx_n$.
\end{proof}

The next theorem is a necessary and sufficient condition on the existence
of parallel telescopers for $D$-finite elements.

\begin{thm}\label{THM:existence}
Let~$f_1, \ldots, f_n$ be $D$-finite elements of an $\bbR$-module. 
Then they have a parallel telescoper with respect to~$\vx$
if and only if there exists a nonzero operator~$P {\in} k(t)\lrD$ such that
\begin{equation} \label{EQ:pcc}
P(D_{i}(f_j) - D_{j}(f_i)) = 0\quad \text{for all~$1\leq i<j\leq n$}.
\end{equation}
\end{thm}
\begin{proof}  Assume that~$f_1, \ldots, f_n$ have a parallel telescoper~$P$ with respect to~$\vx$.
Then there exists an element~$g$ in the submodule generated by~$f_1$, \ldots, $f_n$
such that~$P(f_i) = D_{i}(g)$ for all~$i$ with~$1 \le i \le n$. Since~$D_{i}D_{j}(g) = D_{j}D_{i}(g)$,
we have~$P(D_{i}(f_j) - D_{j}(f_i)) = 0$.

Conversely, assume that there exists a nonzero operator $P\in k(t)\lrD$ such that
\[P(D_{i}(f_j) - D_{j}(f_i)) = 0\quad \text{for all~$1\leq i<j\leq n$}.\]
Then~$P(f_1), \ldots, P(f_n)$ are compatible, because~$P$ is free of~$\vx$.  So there is a parallel telescoper~$L$
for~$P(f_1), \ldots, P(f_n)$ by Lemma~\ref{LM:suffcond}. Therefore,~$LP$ is a parallel telescoper for~$f_1,$ \ldots, $f_n$
with respect to~$\vx$.
\end{proof}
%


\section{Hyperexponential case}\label{SECT:hyper}

Let~$\bE$ be a differential field extension of~$(\bK, \Delta)$. The set of extended derivations
on~$\bE$ is also denoted by~$\Delta$. The derivations in~$\Delta$ are assumed to commute with each other.
Furthermore, we assume the subfield of constants in~$\bE$ is~$k$.

For an element~$h \in \bE$ and an operator~$L \in \bbR$ of the form
\[ L = \sum_{i, j_1, \ldots, j_n\geq 0} a_{i,j_1, \ldots, j_n} D_t^i D_1^{j_1} \cdots D_n^{j_n} \]
with~$a_{i,j_1, \ldots, j_n} \in \bK$,  we define the application of~$L$ to~$h$ as
\[ L(h) = \sum_{i, j_1, \ldots, j_n \geq 0} a_{i,j_1, \ldots, j_n} \delta_t^i \circ \delta_1^{j_1} \circ \cdots \circ \delta_n^{j_n}(h).  \]
Then~$\bE$ is an $\bbR$-module whose multiplication is the application of an operator in~$\bbR$ to an element of~$\bE$.

A nonzero element~$h \in \bE$ is said to be {\em hyperexponential} over~$\bK$ if the logarithmic derivative~$\delta(h)/h$ belongs to~$\bK$
for all~$\delta \in \Delta$.  Hyperexponential functions are $D$-finite elements.
In fact, the submodule generated by several hyperexponential functions
over~$\bbR$ is the linear space spanned by them. Two hyperexponential functions are said to be {\em similar} if their ratio belongs to~$\bK$.

\subsection{Determining the existence}\label{subsec:hypexpexistence}
The next proposition  allows one to determine the existence of parallel telescopers for
hyperexponential functions.
\begin{prop} \label{PROP:exp}
Let~$h \in \bE$ be hyperexponential over~$\bK$.
Then~$\ann(h) \cap k(t) \langle D_t \rangle \neq \{0\}$ if and only if
the logarithmic derivative of~$h$ with respect to~$t$ is of the form
\begin{equation}  \label{EQ:sep}
\frac{\delta_t(p)}{p} + r \quad \text{ for some~$p \in k(\vx)[t]$ and~$r \in k(t)$}.
\end{equation}
\end{prop}
\begin{proof}
Assume that~$\delta_t(h)/h$ is of the form~\eqref{EQ:sep}.
Since~$p$ is a polynomial in~$t$ over~$k(\vx)$, there exists a nonzero operator~$L$
in~$k(t)\langle D_t \rangle$ annihilating~$p$.
It is easy to verify that~$(D_t - r)(h/p)=0$.  Therefore,~$h$ is annihilated  by a nonzero
operator in~$k(t)\langle D_t \rangle$. Such an operator is
the symmetric product of~$L$ and~$D_t - r$.

Conversely, assume that there exists a nonzero element~$L{\in}\ann(h) \cap k(t) \langle D_t \rangle$.  Then~$\delta_t(h)/h$
is a rational solution of the Riccati equation associated to~$L$, although it does not have to be in~$k(t)$. By formula~(4.3) in~\cite[page 107]{VdputSinger2003},
\[ \frac{\delta_t(h)}{h} = \frac{\delta_t(P)}{P} + Q + \frac{R}{S}, \]
where~$P, Q, R$ and~$S$ are polynomials in~$t$ over the algebraic closure of~$k(\vx)$, the roots of~$S$ are singular points of~$L$,
and the roots of~$P$ are nonsingular ones (see also~\cite[Theorem~1]{Bronstein1992}). Moreover, one can assume that~$\deg_t(R) < \deg_t(S)$ and that~$S$ is monic. Since the singular points of~$L$
are in~$k$, the coefficients of~$S$ are in~$k$ as well.
Following the algorithm for computing rational solutions of Riccati equations
described in~\cite[{\S}\,4.3]{Bronstein1992} or~\cite[Exercise 4.10]{VdputSinger2003},
we see that~$R$ belongs to~$k[t]$. The same conclusion holds for~$Q$
by the algorithm in~\cite[{\S}\,4.2]{Bronstein1992}, as~$Q$ is constructed by analyzing the pole of the associated Riccati equation at infinity.
Set~$r= Q +  R/S$, which belongs to~$k(t)$, and set~$s=\delta_t(h)/h-r$, which is in~$k(t, \vx)$ and equal to~$\delta_t(P)/P$.
Thus, the linear differential equation~$\delta_t(Y)=s Y$ has a polynomial solution~$P$.
Since~$s$ belongs to~$k(t, \vx)$, the equation must have a polynomial solution~$p$ in~$k(\vx)[t]$,
which implies that~$\delta_t(p)/p = s$. Then, the logarithmic derivative~$\pa h/\pa t$ is of the form~\eqref{EQ:sep}.
\end{proof}

One can decide if the logarithmic derivative~$\delta_t(h)/h$ in Proposition~\ref{PROP:exp} is of the form~\eqref{EQ:sep} by computing its squarefree partial fraction decomposition with respect to~$t$.
A more efficient way is to apply Algorithm \textsf{WeakNormalizer} in~\cite[{\S}\,6.1]{BronsteinBook} to~$\delta_t(h)/h$,
which delivers a polynomial~$p$ in~$k(\vx)[t]$ such that the difference of~$\delta_t(h)/h$ and~$\delta_t(p)/p$
belongs to~$k(t)$ if and only if~$\delta_t(h)/h$
is of the form~\eqref{EQ:sep}.

Let~$h_1, \ldots, h_n$ be hyperexponential functions. Then~$h_1,$ \ldots, $h_n$
have a parallel telescoper with respect to~$\vx$ if and only if, for every pair~$i, j$ with~$1 \le i < j \le n$,
there exists a nonzero operator
$P_{i,j} \in k(t)\langle D_t
\rangle$ such that
\begin{equation}\label{EQ:ij}
P_{i,j}\left(D_i(h_j)-D_j(h_i) \right) = 0.
\end{equation}
This is because the least common left multiple of the~$P_{i,j}$ can be taken as the operator~$P$ in~\eqref{EQ:pcc} of Theorem~\ref{THM:existence}.
For each pair~$(h_i, h_j)$, there are three cases to be considered: (i) If~$D_i(h_j)=D_j(h_i)$, then set~$P_{i, j}=1$.
(ii) If~$h_i$ is similar to~$h_j$, then the difference~$D_i(h_j)-D_j(h_i)$ is hyperexponential. So we can
find~$P_{i, j}$ by Proposition~\ref{PROP:exp}. (iii) If~$h_i$ is not similar to~$h_j$,
then~\eqref{EQ:ij} implies that both~$P_{i,j}(D_i(h_j))$ and~$P_{i, j}(D_j(h_i))$ are equal to zero.
Proposition~\ref{PROP:exp} is also applicable to the last case.

\begin{example}\label{Exam:existence}
Consider the hyperexponential functions
\[h_1 = \frac{t(x_1{+}t{+}t^2u)}{u(t+x_1)\sqrt{t}}, \,\,\,  h_2 = \frac{((t{+}1)^2{+}x_1x_2{+}t(x_1{-}1))u{-}tx_1}{u(t+x_2)\sqrt{t}},\]
where~$u := t+x_1+x_2$. A direct calculation yields
\[h := D_2(h_1) - D_1(h_2) =-\frac{1}{\sqrt{t}} \]
The logarithmic derivative of~$h$ in~$t$ belongs to~$k(t)$.
Then $P := 2tD_t+1$ is the operator in~$k\lrD$ such that~$P(h)=0$.
So~$h_1$ and~$h_2$ have a parallel telescoper with respect to~$x_1$ and~$x_2$ by Proposition~\ref{PROP:exp}.
\end{example}

\subsection{Computing minimal parallel telescopers}\label{subsec:computingparatele}

This subsection is devoted to computing minimal parallel telescopers.
First, we present
a recursive algorithm, named \textsf{ParaTele},  for hyperexponential functions
that are both compatible and similar.
Next, we show that the algorithm can
be easily adapted to compute
minimal parallel telescopers for merely compatible hyperexponential functions.

\smallskip \noindent
{{\bf Algorithm}~\textsf{ParaTele}}:~Given compatible functions~$r_1h,$ \ldots, $r_nh$,
where~$h$ is hyperexponential over~$\bK$ and $r_1$, \ldots, $r_n$ are rational functions in~$\bK$,
compute a minimal parallel telescoper~$L(t, D_t)$ for
$r_1 h,  \ldots, r_n  h$ with respect to~$\vx$ and a certificate~$g$ of~$L$.
\begin{enumerate}
\item Compute a minimal telescoper~$L_n$ for~$r_n h$ with
certificate~$g_n$ by the algorithms in~\cite{Almkvist1990, BCCLX2013}.
\item If~$n=1$, then return~$(L_n, g_n)$; otherwise, set
\[\text{$\tilde f_i := D_i(g_n) - L_n(r_i h)$ for~$i = 1, \ldots, n-1$.}\]
\item \label{Step3} Run {\textsf{ParaTele}} for functions~$\tilde f_1, \ldots, \tilde f_{n-1}$
to get $(\tilde L,  \tilde g)$,
where~$\tilde L$ is of minimal order and $\tilde g$ is in the submodule
generated by $\tilde f_i$'s
over the ring $k(t, x_1, {\ldots}, x_{n-1})\langle D_1, {\ldots}, D_{n-1} \rangle$.
\item Return~$L := \tilde L L_n$ and~$g := \tilde L (g_n) - \tilde g$.
\end{enumerate}


Note that some of the rational functions~$r_1$, \ldots,~$r_n$ in Algorithm~\textsf{ParaTele}
may be equal to zero. So the input consists of either zero or similar hyperexponential functions.
This guarantees that the recursion in step~$3$ can be executed.

It follows from the proof of Lemma~\ref{LM:suffcond} that Algorithm \textsf{ParaTele} always computes
 a parallel telescoper. To show its minimality, we need a lemma that
plays a similar role for hyperexponential functions as Lemma~\ref{LM:eliminate} for $D$-finite ones.

Recall that~$\vx_{m}=(x_1, \ldots, x_{m})$ and~$\bbR_{m}$
denotes $k(t, \vx_{m})\langle D_t, D_1, \ldots, D_{m} \rangle$, where~$m=1, \ldots, n$.
\begin{lem} \label{LM:var}
Let~$h_1, \ldots, h_{m}$ be hyperexponential elements of~$\bE$.
Assume that, for all~$i$ with~$1 \le i \le m$,
\begin{equation} \label{EQ:zero}
D_{m+1}(h_i)=  \cdots = D_n(h_i)=0.
\end{equation}
Let~$N$ and~$N_m$ be the submodule generated by~$h_1$, \ldots, $h_{m}$ over~$\bbR$ and~$\bbR_m$, respectively.
If there exists a nonzero operator~$T \in k(t)\lrD$ and~$a \in N$ such that
$$T(h_i)=D_i(a) \quad \text{for all~$i$ with~$1 \le i \le m,$} $$
 then there exists~$b \in N_m$
such that
$$T(h_i)=D_i(b) \quad \text{for all~$i$ with~$1 \le i \le m$}.$$
In other words,~$T$ is a parallel telescoper for~$h_1$, \ldots,~$h_m$ with respect to~$\vx_{m}$.
\end{lem}
\begin{proof}
Without loss of generality, assume that~$\{h_1, \ldots, h_\ell\}$ is a maximal
linearly independent subset of~$\{h_1, \ldots, h_m\}$ over~$\bK$. Then~$a = \sum_{j=1}^\ell a_j h_j$
for some~$a_j \in \bK$, because~$N$ is
the linear space spanned by~$h_1$, \ldots,~$h_\ell$ over~$\bK$. Hence,
\begin{equation} \label{EQ:T1}
T(h_i) = \sum_{j=1}^\ell D_i(a_jh_j) = \sum_{j=1}^\ell \left( \delta_i(a_j) + a_j r_{i,j} \right) h_j,
\end{equation}
where~$r_{i,j}$ stands for the logarithmic derivative~$\delta_j(h_i)/h_i$ and~$i$ ranges from~$1$ to~$m$.
Then there exist~$s_i \in \{1, \ldots, \ell\}$ and~$w_{i,s_i} \in k(t,\vx_m)$
such that
$$T(h_i)=w_{i,s_i}h_{s_i} \quad \text{for all~$i \in \{1, \ldots, m\}$}.$$
In fact,~$s_i$ can be any integer between~$1$ and~$\ell$ and~$w_{i,s_i}$ must be zero if~$T(h_i)=0$; and~$s_i$ is unique
if~$T(h_i)$ is nonzero by Proposition~4.1 in~\cite{LiWuZheng2007}. Thus,~\eqref{EQ:T1} can be rewritten as
$$w_{i,s} h_s = \sum_{j=1}^\ell \left( \delta_i(a_j) + a_j r_{i,j} \right) h_j.$$
By the linear independence of~$h_1$, \ldots,~$h_\ell$, $T(h_i) = D_i(a)$ is equivalent to
\begin{equation} \label{EQ:cond}
\left\{ \begin{array}{l}
\delta_i(a_{s_i}) + a_s r_{i,{s_i}} = w_{i,s_i}, \\ \\
\delta_i(a_j) + a_j r_{i,j}=0~\text{for~$j {\in} \{1, \ldots, m \}$ with~$j {\neq} s_i$.}
\end{array} \right.
\end{equation}

Let~$\xi_{m+1}, \ldots, \xi_n \in k$ be such that~$b_i = a_i(\vx_m, \xi_{m+1}, \ldots \xi_n)$
is well-defined for all~$i$ with~$1 \le i \le \ell.$ Then~\eqref{EQ:cond}
still holds if we replace~$a_i$ by~$b_i$ for~$i=1$, \ldots, $m$.  This is because the
substitution of~$\xi_{m+1},$ \ldots,~$\xi_n$ for~$x_{m+1}$, \ldots,~$x_n$ commutes with~$\delta_i$ for
all~$i$ with~$1 \le i \le m$; and because both~$w_{i,s_i}$ and the~$r_{i,j}$'s are free of~$x_{m+1}$, \ldots,~$x_n$
by~\eqref{EQ:zero}.

Set~$b=\sum_{j=1}^\ell b_j h_j$, which
is in~$N_m$. It follows from~\eqref{EQ:cond} that~$T(h_i)=D_i(b)$ for all~$i$
with~$1 \le i \le m$.
\end{proof}
%
%

We now prove the correctness of Algorithm~\textsf{ParaTele}.
\begin{prop} \label{PROP:minimal}
Let~$h \in \bE$ be hyperexponential over~$\bK$ and~$r_1, \ldots, r_n \in \bK$.
If~$r_1 h, \ldots, r_n h$ are compatible,
then Algorithm~\textsf{ParaTele} computes a minimal parallel telescoper
for~$r_1 h, \ldots, r_n h$ with respect to~$\vx$.
\end{prop}
\begin{proof}
Set~$f_i = r_i h$ for~$i=1$, \ldots,~$n$.
Note that~$L_n$ and~$g_n$ obtained from step~1 can be identified with the telescoper and certificate
in~\eqref{EQ:teleM}, respectively, because the $\bbR$-submodule generated by~$f_1$, \ldots,~$f_n$
is equal to that generated by~$h$. Consequently,
Algorithm~\textsf{ParaTele} is just an algorithmic formulation
of the proof of Lemma~\ref{LM:suffcond} with an additional assumption
that~$\tilde L$ is a minimal parallel telescoper for~$\tilde f_1, \ldots,~\tilde f_{n-1}$ with respect to~$\vx_{n-1}$.
The conclusions made in the proof of Lemma~\ref{LM:suffcond} remain valid. In particular,~$\tilde L L_n$
is a parallel telescoper for~$f_1$, \ldots,~$f_n$ with respect to~$\vx$.

It remains to prove that~$\tilde L L_n$ is of minimal order.
Assume that~$P \in k(t)\lrD$ is a parallel telescoper
for~$f_1$, \ldots,~$f_n$ with respect to~$\vx$. Then~$P(f_i)=D_i(w)$ for all~$i$ with~$1 \le i \le n$
and for some~$w$ in the submodule generated
by~$f_1$, \ldots,~$f_n$ over~$\bbR$.
In particular,~$P$ is a telescoper for~$f_n$ with respect to~$x_n$.
Thus,~$P=QL_n$ for some~$Q \in k(t) \lrD$.
Applying~$Q$ to~$\tilde f_1, \ldots,~\tilde f_{n-1}$ yields
\begin{equation} \label{EQ:teleN}
Q(\tilde f_i) = Q D_i(g_n) - P(f_i) = D_i(Q(g_n)-w)
\end{equation}
for all~$i$ with~$1 \le i \le n-1$. By~\eqref{EQ:const} in the proof of Lemma~\ref{LM:suffcond},~$D_n(\tilde f_i)=0$
for all~$i$ with~$1 \le i \le n-1$. So Lemma~\ref{LM:var} implies that~$Q$ is a parallel telescoper
for~$\tilde f_1, \ldots,~\tilde f_{n-1}$ with respect to~$\vx_{n-1}$.
Thus,~$Q$ is a left multiple of~$\tilde L$ obtained in step~3,
because~$\tilde L$
is a minimal parallel telescoper for~$\tilde f_1, \ldots,~\tilde f_{n-1}$. So~$P$ is a left multiple of
the product~$\tilde L L_n$. Thus,~$\tilde L L_n$ is of minimal order.
\end{proof}

\begin{example}
Let~$h_1, h_2$ and~$P$ be the same as in Example~\ref{Exam:existence}.
Then~$H_1 := P(h_1)$ and~$H_2 := P(h_2)$ are compatible hyperexponential
functions.  Applying any telescoping algorithm
in~\cite{Almkvist1990, BCCLX2013} to~$H_1$ yields a minimal telescoper
\[L_1:=4t^2D_t^2-8tD_t+5\]
for~$H_1$ satisfying~$L_1(H_1) = D_1(G_1)$ for some hyperexponential function~$G_1$ over~$k(t, x_1, x_2)$.
Set
\[ \tilde{H}_2 := L_1(H_2) - D_{2}(G_1) =  -\frac{16t^2(4tx_2^2+4t+x_2^3+x_2)}{(x_2+t)^4\sqrt{t}}\]
A minimal telescoper for~$\tilde{H}_2$ is~$L_2 := 2tD_t-3$.
Then
\[L=L_2L_1 := 8t^3D_t^3-12t^2D_t^2+18tD_t-15\]
is a minimal parallel telescoper for~$H_1$ and~$H_2$ with respect to~$x_1$ and~$x_2$.
\end{example}

Let us consider how to compute a minimal telescoper for compatible hyperexponential functions~$f_1, \ldots,~f_n$.
For simplicity, we assume that~$f_1$, \ldots, $f_m$ and~$f_{m+1}$, \ldots,~$f_n$ form two distinct equivalence classes
modulo similarity. The same idea applies to the case, in which there are more than two equivalence classes.
Since~$f_i$ and~$f_j$ are not similar for all~$i$ with~$1 \le i \le m$ and~$j$ with~$m+1 \le j \le n$, the compatibility
condition~$D_i(f_j)=D_j(f_i)$ implies that~$D_i(f_j){=}D_j(f_i){=}0$. Let~$P$ be a minimal telescoper for~$f_1$, \ldots,~$f_m$
over~$\vx_m$, and~$Q$ a minimal one for~$f_{m+1}$, \ldots,~$f_n$ with respect to~$x_{m+1}$, \ldots,~$x_n$.
Then~$P(f_i)=D_i(g)$ for all~$i$ with~$1 \le i \le m$ and for some~$g$ in the submodule generated by~$f_1$, \ldots,~$f_m$
over~$\bbR_m$,
and~$Q(f_j)=D_j(h)$ for all~$j$ with~$m{+}1 {\le} j {\le} n$ and for some~$h$ in the submodule generated by~$f_{m+1}$,
\ldots,~$f_n$ over~$k(t, x_{m+1}, \ldots, x_n)\langle D_t, D_{m+1}, \ldots, D_n \rangle$.
In particular, we have~$D_i(h)=D_j(g)=0$ for all~$i$ with~$1 \le i \le m$ and~$j$ with~$m+1 \le j \le n$.

Set~$L$ to be the least common left multiple of~$P$ and~$Q$. Then there exist~$U, V \in k(t)\lrD$ such
that~$L=UP=VQ$. A straightforward calculation implies that~$L$ is a parallel telescoper for~$f_1$, \ldots, $f_n$
with respect to~$\vx$. A certificate of~$L$ is~$U(g)+V(h)$. Let~$L^\prime$ be a parallel telescoper
for~$f_1$, \ldots,~$f_n$ with respect to~$\vx$. By Lemma~\ref{LM:var},~$L^\prime$ is a parallel telescoper
for both~$f_1$, \ldots,~$f_m$ with respect to~$\vx_m$ and~$f_{m+1}$, \ldots,~$f_n$ with respect to~$x_{m+1}$,
\ldots,~$x_n$. So it is a common left multiple of~$P$ and~$Q$. Consequently, it is a left multiple of~$L$.
We conclude that~$L$ is a minimal telescoper for~$f_1$, \ldots,~$f_n$ with respect to~$\vx$.

To construct a parallel telescoper for hyperexponential functions~$f_1, \ldots, f_n$
that are not necessarily compatible with respect to~$\vx$, we
compute a nonzero operator~$P \in k(t)\lrD$ such that~\eqref{EQ:pcc} holds.
Then~$P(f_1)$, \ldots,~$P(f_n)$ are compatible with respect to~$\vx$.
Let~$L$ be a parallel telescoper for~$P(f_1)$, \ldots,~$P(f_n)$.
Then~$LP$ is a parallel telescoper for~$f_1$, \ldots,~$f_n$. But~$LP$
is not necessarily of minimal order.

\section{Parameterized Picard--Vessiot Theory}\label{SECT:ppv}
A generalized differential Galois theory having
differential algebraic groups (as in \cite{Kolchin1985}) as Galois groups was initiated in \cite{Landesman2008}.
The parameterized Picard--Vessiot theory considered in \cite{CassidySinger2007} is a special case of the above generalized differential Galois theory and
studies symmetry groups of the solutions of linear differential equations whose coefficients contain parameters. In this section, we show the connection of
parallel telescoping with this parameterized theory.

Let~$F$, containing~$k(t)$ as a subfield, be a differentially closed field of characteristic zero, i.e.,
any consistent differential system with coefficients in~$F$ has solutions in~$F$.
Let~$F(\vx)$ be the field of rational functions in~$\vx$.  As before,~$\Delta$  stands for the
set~$\{\delta_t, \delta_{1}, \ldots, \delta_{n}\}$ of derivations, and~$E$ is an~$\bbR$-module
as described at the beginning of  Section~\ref{SECT:hyper}.

Let~$E$ be a differential field extension of~$F(\vx)$.
For a subset~$\Lambda\subset \Delta$, an element~$c\in E$ is called a \emph{$\Lambda$-constant} if~$\lambda(c)=0$  for all~$\lambda \in \Lambda$. The set of all $\Lambda$-constants forms a subfield of~$E$, which is denoted by~$C_E^{\Lambda}$.
Consider the differential system
\begin{equation}\label{EQ:psys}
D_{1}(Y) = A_1 Y, \, \,  \ldots, \, \, D_{n}(Y) = A_n Y,
\end{equation}
where~$A_i\in \gl_n(F(\vx))$, the set of~$n\times n$ matrices with entries in~$F(\vx)$, such that
\[D_{i}(A_j) - D_{j}(A_i) = A_iA_j - A_jA_i.\]
As in the classical Galois theory, we now define the \lq\lq splitting field\rq\rq~for the system~\eqref{EQ:psys}.
\begin{define}\label{DEF:ppv}
A \emph{parameterized Picard--Vessiot extension of~$F(\vx)$} (abbreviated as PPV-extension of~$F(\vx)$) for the
system ~\eqref{EQ:psys} is a $\Delta$-field extension~$E$ of~$F(\vx)$ satisfying
\begin{itemize}
\item[(a)] There exists a matrix~$Z\in {\rm GL}_n(E)$ such that~$D_{i}(Z)=A_iZ$ for all~$i=1, \ldots, n$
and~$E$ is generated as a $\Delta$-field over~$F(\vx)$ by the entries of~$Z$.
\item[(b)]  $C_E^{\Lambda} = C_{F(\vx)}^{\Lambda} = F$ for~$\Lambda = \{\delta_1, \ldots, \delta_n\}$.
\end{itemize}
The \emph{parameterized Picard--Vessiot group} (abbreviated as PPV-group) associated with the PPV-extension~$E$ of~$F(\vx)$ is the group
\[ \Gal_{\Delta}(E/F(\vx)) = \text{$\{\sigma\in {\rm Aut}_{F(\vx)}(E) \mid \sigma \delta = \delta \sigma$ for $\delta \in \Delta$\}}.\]
\end{define}
The existence of PPV-extensions for parameterized differential systems has been established
in~\cite[Theorem 9.5~(1)]{CassidySinger2007} under the assumption that~$F$ is differentially closed.
Recently, this existence result has been improved so that one only needs~$F$ to be algebraically closed~\cite{Wibmer2012}
and under weaker closure conditions in~\cite{GGO2013}. In the classical Galois theory,
the Galois group of an algebraic equation is a subgroup
of the permutation group. In the non-parameterized differential case, the Galois group
of a linear differential system is a linear algebraic group, i.e., a group of~$n\times n$ matrices whose entries are elements in the field of constants satisfying certain polynomial equations.
The PPV-group associated with a PPV-extension of~$F(\vx)$ is a \emph{linear differential algebraic group}, i.e., a group of~$n\times n$ matrices whose entries are elements in~$F$
satisfying certain differential equations~\cite[Theorem 9.5~(2)]{CassidySinger2007}.

\begin{example}[Example 3.1 in~\cite{CassidySinger2007}] \label{EXAM:ppv1}
Consider the equation
\[D_x(Y) = \frac{t}{x}\, Y.\]
The PPV-extension for this equation is the $\{\delta_t, \delta_x\}$-field, generated by the element~$z=x^t$, i.e.,
\[E \triangleq F(x)(z, \delta_x(z), \delta_t(z),  \ldots  )= F(x, x^t, \log(x)).\]
The corresponding PPV-group is as follows:
\begin{align*}
  \Gal_{\Delta}(E/F(x)) & = \{ a\in F \mid a\neq 0\,\, \text{and}\,\, \delta_t\left(\frac{\delta_t(a)}{a}\right)=0\}.
\end{align*}
\end{example}

As a corollary of the general Galois correspondence~\cite[Theorem 9.5]{CassidySinger2007}, the following
lemma will be used frequently in the rest of this section.

\begin{lem}\label{LM:galois}
Let~$E$ be a PPV-extension of~$F(\vx)$ for some parameterized differential system
and let~$\Gal_{\Delta}(E/F(\vx))$ be the associated PPV-group. Then the set
\[\{f\in E \mid \sigma(f)=f\, \text{for all~$\sigma\in \Gal_{\Delta}(E/F(\vx))$}\} \]
coincides with the field~$F(\vx)$.
\end{lem}

\subsection{Galois groups of first-order systems}

Unlike the usual Picard--Vessiot theory where we have a complete algorithm to compute the Galois group of a given linear differential equation over the field of rational functions~\cite{Hrushovski2002}, we have only partial algorithmic results for the PPV-theory.
Algorithms for first and second order parameterized equations over $F(\vx),$ where~$n=1$, appear
in~\cite{Arreche2012,Dreyfus2011}.
An algorithm to determine if a parameterized equation of arbitrary order has
a unipotent PPV-group (or even certain kinds of extensions of such a group) as well
as an algorithm to compute the group appears in~\cite{MOS2013b}. An algorithm to determine
if a parameterized equation has a reductive PPV-group and compute it if it does appears in~\cite{MOS2013a}.

We now show how one determines the PPV-group of a first-order differential system
of the form
\begin{equation}\label{EQ:fosys}
D_{1}(Y) = f_1, \, \,  \ldots, \, \, D_{n}(Y) = f_n,
\end{equation}
where~$f_1, \ldots, f_n\in F(\vx)$ are compatible rational functions with respect to~$\vx$.
Let~$E$ be the PPV-extension  of~$F(\vx)$ and let~$z\in E$ be a solution of the system~\eqref{EQ:fosys}.
For every~$\sigma \in \Gal_{\Delta}(E/F(\vx))$, $\sigma(z)$ is still a solution of the system~\eqref{EQ:fosys}.
Then $\sigma(z) {=} z {+} c_{\sigma}$ for some~$c_{\sigma} \in C_E^{\Lambda}=F$
with~$\Lambda = \{\delta_{1}, \ldots, \delta_{n}\}$.  By fixing a solution~$z$, we get a representation of
the PPV-group~$\Gal_{\Delta}(E/F(\vx))$ as a subgroup of the additive group~$(F, +)$.
The subgroups of~$(F, +)$ have been classified by Cassidy~\cite[Lemma 11]{Cassidy1972} and Sit~\cite[Theorem 1.3, p.647]{Sit1975}. That is, any subgroup~$G$
of~$(F, +)$ is of the form~$\{a\in F \mid L(t, D_t)(a)=0\}$,
where~$L$ is a linear differential operator in~$F\lrD$.
We call~$L$ the~\emph{defining operator} for~$G$.

\begin{lem}\label{LM:ppvcoeff}
If the coefficients~$f_1, \ldots, f_n$ of the system~\eqref{EQ:fosys} are in~$k(t)(\vx)$, then the defining operator~$L$
for its corresponding PPV-group is in~$k(t)\lrD$. \end{lem}
\begin{proof}
As noted above PPV-group~$G$  can be identified with the set of solutions of
an equation of the form~$L(y) {=} 0$ where~$L\in F\langle D_t\rangle$.
To see that this group is actually defined over~$k(t)$, note that~$k(t, \vx)$ is a
purely transcendental extension of~$k(t)$ and so~$k(t)$ is algebraically closed in~$k(t, \vx)$.
Furthermore,~$\Delta$ consists of independent derivations over~$k(t, \vx)$. Remark~2.9.2 and Theorem 2.8 of \cite{GGO2013} imply that a parameterized
Picard-Vessiot extension exists for our equations and Lemma~8.2 of \cite{GGO2013} implies that the parameterized Picard-Vessiot group is defined over $k(t)$. Finally~\cite[Theorem 1.3, p.647]{Sit1975} implies that this group is defined as claimed above.
\end{proof}

Now we present the main result of this section that the problem of determining the PPV-group of
the system~\eqref{EQ:fosys} with coefficients in~$k(t, \vx)$ is equivalent to that of computing a minimal parallel telescoper for its coefficients.

\begin{thm}\label{THM:direct}
Let~$f_1, \ldots, f_n$ be the coefficients of the system~\eqref{EQ:fosys} such that they are in~$k(t, \vx)$
and compatible with respect to~$\vx$. Then~$L\in k(t)\lrD$ is the defining operator for the PPV-group of the system~\eqref{EQ:fosys}
if and only if $L$ is a minimal parallel telescoper for $f_1, \ldots, f_n$ with respect to~$\vx$.
\end{thm}
\begin{proof}
Let~$\tilde{L}$ be the defining operator for the PPV-group~$G$ of the system~\eqref{EQ:fosys}.
By Lemma~\ref{LM:ppvcoeff}, $\tilde{L}$ is in~$k(t)\lrD$.
We claim that~$\tilde L$ is a parallel telescoper for
$f_1, \ldots, f_n$ with respect to~$\vx$.
Let~$z\in E$ be a solution of~\eqref{EQ:fosys}.
Then, for any~$\sigma\in G$, $\sigma(z) = z + c_{\sigma}$, where~$c_{\sigma}\in F$ is such that~$\tilde L(c_{\sigma}) =0$.  For any~$\sigma\in G$, $\sigma(\tilde L(z)) = \tilde L(\sigma(z)) = \tilde L(z+c_{\sigma}) = \tilde L(z) + \tilde L(c_{\sigma}) = \tilde L(z)$. Then~$\tilde g:=\tilde L(z)\in F(\vx)$ by Lemma~\ref{LM:galois}. Since~$\tilde L$
commutes with~$D_{i}$ for all~$i=1, \ldots, n$,
\begin{equation}\label{EQ:ptele}
\tilde L(D_{i}(z)) = \tilde L(f_i) = D_{i}(\tilde g).
\end{equation}
We now show that we can choose~$\tilde{g}\in k(t,\vx)$~(and so~$\tilde{L}$
will be a parallel telescoper with~$\tilde{g}\in k(t,\vx)$).
Since the~$f_i$~are in~$k(t,\vx)$, equations~(\ref{EQ:ptele}) imply that~$D_i(\tilde{g})$ belongs to~$k(t,\vx)$.
Expanding $\tilde g$ in partial fractions
with respect to~$x_n$ and using induction, one sees that there is an element $c\in F$ such that~$\tilde g - c \in k(t, \vx)$.
Now let~$L\in k(t)\lrD$ be a minimal parallel telescoper for $f_1, \ldots, f_n$ with respect to~$\vx$.
Then, we have that~$L$ divides~$\tilde L$. To complete the proof,  it remains to show that~$\tilde L$ divides~$L$.
It suffices to prove that $L(c_{\sigma}) = 0$ for any~$\sigma\in G$.
Since~$L(f_i)= D_{i}(g)$ for some~$g\in k(t, \vx) \subset F(\vx)$,
$$ D_{i}(L(z)-g) =L(D_{i}(z)) - D_{i}(g)= L(f_i) - D_{i}(g) = 0.$$
Therefore~$L(z)-g \in F$ and so~$L(z)\in F(\vx)$.  For any~$\sigma\in G$, $L(z) = \sigma(L(z)) = L(\sigma(z)) = L(z + c_{\sigma})=L(z) + L(c_{\sigma})$. Then~$L(c_{\sigma}) = 0$.
\end{proof}

\begin{example} \label{EXAM:ppv2}
Consider the differential system
\begin{equation}\label{EQ:ppv2}
  D_{1}(Y) = f_1, \quad D_2(Y)=f_2,
\end{equation}
where~$f_1, f_2\in k(t, x_1, x_2)$ are of the form
\[f_1 = \frac{t}{x_1+x_2+t}, \quad f_2 = \frac{tx_2+t^2+x_1+x_2+t}{(x_1+x_2+t)(x_2+t)}.\]
It is easy to check that~$D_{2}(f_1)=D_{1}(f_2)$. Applying any telescoping algorithm
in~\cite{Almkvist1990, BCCL2010} to~$f_1$ yields a minimal telescoper~$L_1$ for~$f_1$
with certificate~$g_1$ as follows
\[L_1 = tD_t - {1}\quad \text{and}\quad g_1 =\frac{t^2}{x_1+x_2+t} .\]
Set~$\tilde{f}_2 := L_1(f_2) - D_{2}(g_1) =  -  (2t+x_2)/(x_2+t)^2$.
A minimal telescoper for~$\tilde{f}_2$ and its certificate are
\[L_2 = D_t \quad \text{and}\quad g_2 = \frac{-t}{(x_2+t)^2}.\]
Then~$L=L_2L_1=tD_t^2$ is a minimal parallel telescoper for~$f_1$ and~$f_2$.
Then the corresponding PPV-group of system~\eqref{EQ:ppv2} is as follows:
\[ \Gal_{\Delta}(E/F({\vx})) = \left\{a\in F  \mid  \delta_t^2(a)=0 \right\}.
\]
\end{example}

\subsection{An inverse problem}\label{subsec:inverse}
As in the classical Galois theory, a natural question that arises is the inverse problem: {\it Which groups occur as Galois groups over a given field?}
In \cite[Example 7.1]{CassidySinger2007}, the authors consider a~$\Delta$-field~$F(x)$,  where~$\Delta = \{\delta_t, \delta_x\}$ and~$F$ is a $\{\delta_t\}$-differentially closed field.
They show that the additive group~$(F, +)$  cannot be the Galois group of a parameterized Picard--Vessiot extension of this field.  In rest of this section, we show
a similar result for fields of rational functions in several variables (for other results concerning the inverse problem, see~\cite{Dreyfus2012, MOS2013a, MOS2013b,  MitschiSinger2012, Singer2013}).
The key tool will be the fact that parallel telescopers always exist for compatible rational functions.
Let us first recall a lemma, which is an easy corollary of Theorem 4 of Chapter VII.3 in~\cite{Kolchin1985}.
\begin{lem}\label{LM:kolchin}
If~$G$ is the PPV-group of a PPV-extension $E$ of~$F(\vx)$,
then~$E=F(\vx)\langle z\rangle_{\Delta}$, satisfying for any~$\sigma \in G$
\[\text{$\sigma(z) = z+ c_{\sigma}$ with~$c_{\sigma}\in F$.}\]
\end{lem}

\begin{thm}\label{THM:inverse}
The additive group~$G=(F, +)$ is not the PPV-group of a PPV-extension of~$F(\vx)$.
\end{thm}
\begin{proof}
We argue by contradiction.
Assume that~$G$ is the PPV-group of some PPV-extension~$E$ of~$F(\vx)$.
Then Lemma~\ref{LM:kolchin} implies that~$D_{i}(z) = f_i$ with~$f_i\in F(\vx)$ for all~$i=1,$ \ldots,~$n$.
Since~$D_{i}$ and~$D_{j}$ commute in~$F(\vx)$, $f_1, \ldots, f_n$ are compatible.
By Theorem~\ref{THM:existence}, there exists~$L$ in~$F\lrD$ such that
$L(f_i) = D_{i}(g)$ for some~$g\in F(\vx)$.
By the same argument in the proof of Theorem~\ref{THM:direct}, $L(c_{\sigma}) = 0$ for all~$\sigma\in G$.
This implies that~$G \subset \{c\in F \mid L(c)=0\} \varsubsetneq F$, which is a contradiction with~$G=F$.
\end{proof}

\def\polhk#1{\setbox0=\hbox{#1}{\ooalign{\hidewidth
  \lower1.5ex\hbox{`}\hidewidth\crcr\unhbox0}}}

%
%

\end{document}